\documentclass[11pt,twoside]{article}
\usepackage[margin=1in]{geometry}
\usepackage[mathletters]{ucs}		
\usepackage[utf8x]{inputenc}

\usepackage{microtype}
\usepackage{float,graphicx,verbatim,fullpage}
\usepackage{latexsym,amssymb,amsmath,amsthm}
\usepackage{enumerate,enumitem,multicol,xspace}
\usepackage[nocompress]{cite}
\usepackage{mathtools}
\usepackage[dvipsnames,usenames]{color}
\usepackage[colorlinks=true, urlcolor=Blue, citecolor=Green, linkcolor=BrickRed, breaklinks, unicode]{hyperref}
\usepackage[l2tabu, orthodox]{nag}
\usepackage{stmaryrd} 

\def\final{0}  
\def\iflong{\iffalse}
\ifnum\final=0  
\newcommand{\knote}[1]{\note{Karthik: #1}}
\newcommand{\hnote}[1]{\note{Hsien: #1}}
\newcommand{\anote}[1]{\note{Alex: #1}}
\newcommand{\cnote}[1]{\note{Charlie: #1}}
\newcommand{\todonote}[1]{\note{TODO: #1}}
\newcommand{\sidecomment}[1]{}
\else 
\newcommand{\knote}[1]{}
\newcommand{\hnote}[1]{}
\newcommand{\anote}[1]{}
\newcommand{\cnote}[1]{}
\newcommand{\todonote}[1]{}
\newcommand{\sidecomment}[1]{}
\fi  

\usepackage[mathlines, pagewise]{lineno}

\newcommand*\patchAmsMathEnvironmentForLineno[1]{%
 \expandafter\let\csname old#1\expandafter\endcsname\csname #1\endcsname
 \expandafter\let\csname oldend#1\expandafter\endcsname\csname end#1\endcsname
 \renewenvironment{#1}%
 {\linenomath\csname old#1\endcsname}%
 {\csname oldend#1\endcsname\endlinenomath}}%
\newcommand*\patchBothAmsMathEnvironmentsForLineno[1]{%
 \patchAmsMathEnvironmentForLineno{#1}%
 \patchAmsMathEnvironmentForLineno{#1*}}%
\AtBeginDocument{%
\patchBothAmsMathEnvironmentsForLineno{equation}%
\patchBothAmsMathEnvironmentsForLineno{align}%
\patchBothAmsMathEnvironmentsForLineno{flalign}%
\patchBothAmsMathEnvironmentsForLineno{alignat}%
\patchBothAmsMathEnvironmentsForLineno{gather}%
\patchBothAmsMathEnvironmentsForLineno{multline}%
}

\def\mathsc#1{\text{\textsc{#1}}}

\definecolor{DarkRed}{rgb}{0.50,0.00,0.00}
\def\EMPH#1{{{\emph{#1}}}} 

\newtheorem{theorem}{Theorem}[section]

\newtheorem{lemma}[theorem]{Lemma}
\newtheorem{claim}[theorem]{Claim}
\newtheorem*{claim*}{Claim}
\newtheorem{corollary}[theorem]{Corollary}

\numberwithin{figure}{section}

\theoremstyle{definition}

\def\beginaalgo{\begin{minipage}{1in}\normalfont\begin{tabbing}
        \quad\=\qquad\=\qquad\=\qquad\=\qquad\=\qquad\=\qquad\=\qquad\=\qquad\=\qquad\=\qquad\=\qquad\=\qquad\=\kill}
\def\endaalgo{\end{tabbing}\end{minipage}}

\newenvironment{algorithm}
{\begin{tabular}{|l|}\hline\beginaalgo}
{\endaalgo\\\hline\end{tabular}}

\def\textul#1{\underline{\smash{#1}\vphantom{.}}}

\def\R{\mathbb{R}}

\def\Ceil#1{\left\lceil #1 \right\rceil}

\def\Set#1{\left\{ #1 \right\}}
\def\Abs#1{\left| #1 \right|}

\def\Paren#1{\left( #1 \right)}		

\def\etal{\emph{et~al.}}			

\def\noote#1{\textsf{\boldmath \textbf{$\langle\!\langle$#1$\rangle\!\rangle$}}\leavevmode}
\def\nootewaarn{\GenericWarning{}{AUTHOR WARNING: Unresolved \protect\note}}

\def\nootenoote#1{\marginpar
	[\hfill\llap{\textcolor{red}{{#1}$\!\Longrightarrow$}}]
	{\rlap{\textcolor{red}{$\Longleftarrow\!${#1}}}}}

\def\nootedingboot{\textcircled{$\maltese$}}

\def\note#1{\textcolor{red}{\nootewaarn\noote{\nootenoote{\nootedingboot}#1}}}

\newcommand{\sgn}{\mathrm{sgn}}
\def\cycles{\text{Cycles}}
\def\matchings{\text{Matchings}}
\def\loops{\text{Loops}}

\begin{document}
\title{Invertibility and Largest Eigenvalue of Symmetric Matrix Signings}
\author{Charles Carlson \quad 
Karthekeyan Chandrasekaran  \quad
Hsien-Chih Chang \quad
Alexandra Kolla%
\thanks{Email: \texttt{\{ccarlsn2,karthe,hchang17,akolla\}@illinois.edu}}}
\date{}
\maketitle
\thispagestyle{empty}
\begin{abstract}
The spectra of signed matrices have played a fundamental role in social sciences, graph theory, and control theory. In this work, we investigate the computational problems of identifying symmetric signings of matrices with natural spectral properties. Our results are twofold:
\begin{enumerate}
\item
We show NP-completeness for the following three problems: verifying whether a given matrix has a symmetric signing that is positive semi-definite/singular/has bounded eigenvalues. However, we also illustrate that the complexity could substantially differ for input matrices that are adjacency matrices of graphs.
\item
We exhibit a stark contrast between invertibility and the above-mentioned spectral properties: we show a combinatorial characterization of matrices with invertible symmetric signings and design an efficient algorithm using this characterization to verify whether a given matrix has an invertible symmetric signing.	 
Next, we give an efficient algorithm to solve the search problem of finding an invertible symmetric signing for matrices whose support graph is bipartite.  
We also provide a lower bound on the number of invertible symmetric signed adjacency matrices. Finally, we give an efficient algorithm to find a minimum increase in support of a given symmetric matrix so that it has an invertible symmetric signing. 
\end{enumerate}
We use combinatorial and spectral techniques in addition to classic results from matching theory. 
Our combinatorial characterization of matrices with invertible symmetric signings might be of independent interest. 
\end{abstract}

\newpage
\setcounter{page}{1}
\section{Introduction}
\label{sec:intro}
The spectra of several graph-related matrices such as the adjacency and the Laplacian matrices have become fundamental objects of study in computer science research. They have had a tremendous impact in several areas including machine learning, data mining, web search and ranking, scientific computing, and computer vision. In this work, we undertake a systematic and comprehensive investigation of the spectrum and the invertibility of \emph{symmetric signings} of matrices. 

For a real symmetric $n\times n$ matrix $M$ and an $n\times n$ matrix $s$ taking values in $\{\pm 1\}$---which we refer to as a \EMPH{signing}---we define the \EMPH{signed matrix} $M(s)$ to be the matrix obtained by taking entry-wise products of $M$ and $s$. We say that $s$ is a \emph{symmetric} signing if $s$ is a symmetric matrix and an \emph{off-diagonal} signing if $s$ takes value $+1$ on the diagonal. Signed adjacency matrices (respectively, Laplacians) correspond to the case where $M$ is the adjacency matrix (respectively, Laplacian) of a graph. Signed adjacency matrices were introduced as early as 1953 by Harary \cite{harary}, to model social relations involving disliking, indifference, and liking. They have since been used in an array of network applications such as finding ``balanced groups'' in social networks where members of the same group like or dislike each other\cite{harary} and reaching consensus among a group of competing or cooperative agents \cite{altafini}.
Studies of spectral properties of general signed matrices as well as signed adjacency matrices have led to breakthrough advances in multiple areas such as algorithms \cite{MSS13, Fri08, LP10, BL06, Fri03, ACKM13}, graph theory \cite{seymour, plummer, thomas, vazirani}, and control theory \cite{nemirovski, boyd, rohn1, rohn2, rohn3}. 

In this work, we study natural spectral properties of symmetric signed matrices and address the computational problems of finding/verifying the existence of signings with these spectral properties. We recall that a real symmetric matrix is \emph{positive semi-definite} (psd) if all its eigenvalues are non-negative. We study the following computational problems:

\vspace{11pt}
\begin{itemize}[leftmargin=0cm]
\item[] 
{\textsc{BoundedEvalueSigning}}:  
Given a real symmetric matrix $M$ and a real number $\lambda$, verify if there exists an off-diagonal symmetric signing $s$ such that the largest eigenvalue $\lambda_{\max}(M(s))$ is at most $\lambda$.

\item[] 
{\textsc{PsdSigning}}:  Given a real symmetric matrix $M$, verify if there exists a symmetric signing $s$ such that $M(s)$ is positive semi-definite.

\item[] 
{\textsc{SingularSigning}}: Given a real symmetric matrix $M$, verify if there exists an off-diagonal symmetric signing $s$ such that $M(s)$ is singular.

\item[] 
{\textsc{InvertibleSigning}}: Given a real symmetric matrix $M$, verify if there exists a symmetric signing $s$ such that $M(s)$ is invertible (that is, non-singular).
\end{itemize}
\vspace{11pt}

The magnitude of the eigenvalues of symmetric signed adjacency matrices and Laplacians are fundamental to designing expanders 
and in particular, Ramanujan graphs \cite{MSS13}. Motivated by this application, we examine the problems of finding symmetric signed matrices with bounds on the magnitude of the eigenvalues, that is, \textsc{BoundedEvalueSigning} and \textsc{PsdSigning}. In \textsc{SingularSigning} and \textsc{InvertibleSigning}, we examine the singularity and the invertibility of symmetric signed matrices. Both these properties are completely determined by the determinant. The determinant of signed matrices has a special significance in the computational problems of counting matchings and computing the permanent \cite{seymour, plummer, thomas, vazirani}. In what follows, we elaborate on the motivations behind the above-mentioned decision problems and their search variants, and describe our results. 

\subsection{Connections and Motivations} \label{sec:motivations}
\noindent{\textbf {Determinant of signed graphs.} }
The determinant of signed adjacency matrices of graphs has played an important role in answering fundamental questions concerning graphs and linear systems---e.g., see the survey by Thomas \cite{thomas}. As a highlight, we mention the seminal work of Robertson, Seymour and Thomas \cite{seymour} who gave an efficient characterization of matrices whose permanent can be computed by suitably signing the entries and computing the determinant. In this work we investigate questions concerning the \emph{invertibility} of \emph{symmetric} signed matrices. \\

\noindent{\textbf {Spectra of signed matrices and expanders.}}
A $d$-regular graph is \emph{Ramanujan} if the second eigenvalue $\lambda_2$ of the adjacency matrix of the graph is at most $2\sqrt{d-1}$. Ramanujan graphs have the best \emph{expansion} among regular graphs \cite{Nil91}. For this reason, efficient construction of Ramanujan graphs of all degrees has been an intensely studied area \cite{LPS88,MSS13,cohen}. A combinatorial approach to this problem is through the \emph{$2$-lift} operation that doubles the size of a base graph $G$ while preserving the degree. There is a simple bijection between $2$-lifts and symmetric signed adjacency matrices of $G$. Furthermore, the eigenvalues of the adjacency matrix of a $2$-lift are given by the union of the eigenvalues of the adjacency matrix of $G$ (also called the ``old'' eigenvalues) and the signed adjacency matrix of $G$ that corresponds to the $2$-lift (the ``new'' eigenvalues). 

Marcus, Spielman, and Srivastava \cite{MSS13} showed that the adjacency matrix of every $d$-regular bipartite graph \emph{has} a symmetric signing whose eigenvalues are at most $2\sqrt{d-1}$. Making this result \emph{constructive} would immediately lead to an efficient algorithm to construct bipartite Ramanujan simple graphs of all degrees\footnote{While efficient construction of bipartite Ramanujan \emph{multi-graphs} of all degrees is known \cite{cohen}, it still remains open to efficiently construct bipartite Ramanujan \emph{simple} graphs of all degrees.}. 
However, making the result of Marcus \etal\ \emph{constructive} requires an efficient algorithm to solve the sub-problem: find a symmetric signing for a given $d$-regular bipartite graph for which the largest eigenvalue of the signed adjacency matrix is at most $2\sqrt{d-1}$. This naturally raises the computational problem---does there exist an efficient algorithm to find a symmetric signing that minimizes the largest eigenvalue of a given matrix? Such an algorithm would also solve the sub-problem. This in turn motivates the need to investigate the complexity of \textsc{BoundedEvalueSigning} which is the decision variant of the computational problem.\\

\noindent{\textbf {Randomized algorithms for perfect matching.} }
An algebraic method to verify the existence of a perfect matching in a given {bipartite} graph $G$ is by certifying that the determinant of the \emph{Edmonds matrix} is {not} identically zero \cite{Lov79}. 
The Edmonds matrix of a $n\times n$ bipartite graph $G$ is a $n\times n$ matrix $M$ with variables $x_e$ for the edges such that $M[i,j]:=x_{\{i,j\}}$ if $\{i,j\}$ is an edge of $G$, and $M[i,j]=0$ otherwise. The graph $G$ has a perfect matching if and only if the polynomial $\det(M)$ is not identically zero. A simple way to verify this condition is by substituting uniformly random values from a large enough field $F$ for the variables $x_{e}$ and evaluating the determinant over $F$. In order for this algorithm to succeed with constant probability, the size of the field is required to be $\Omega(n)$, and thus the amount of randomness needed for the algorithm is at least $\Omega(mn)$, where 
$m$ is the number of edges in $G$ \cite{Lov79}.

Consider the following variant of this approach.  Instead of picking a large field $F$, substitute uniformly random $\{\pm 1\}$ values into variables $x_{e}$ and evaluate the determinant of the resulting matrix over the reals. If $G$ has no perfect matching, then the determinant will be zero. In this work, we investigate the probability of such a determinant being non-zero if $G$ has a perfect matching. In particular, we show that the probability is non-zero and moreover, we present a non-trivial lower bound on this probability.\\

\noindent{\textbf {Efficient algorithms for finding the solvability index of a signed matrix.}}
The notion of \emph{balance} of a symmetric signed matrix has been studied extensively in social sciences \cite{He46,harary,HJ80,katai}. 
A signed adjacency matrix is \emph{balanced} if there is a partition of the vertex set such that all edges within each part are positive, and all edges in between two parts are negative (one of the parts could be empty).
A number of works \cite{akiyama,hansen78,sivaraman,katai, zaslavsky1, zaslavsky2}
have explored the problem of minimally modifying signed graphs (or signed adjacency matrices) to convert it into a balanced graph. 

In this work, we introduce a related problem regarding symmetric signed matrices: Given a symmetric matrix $M$, what is the smallest number of non-diagonal zero entries of $M$ whose replacement by non-zeroes gives a symmetric matrix $M'$ that has an invertible symmetric signing? We define this quantity to be the \emph{solvability index}\footnote{Our definition of \emph{solvability index} is similar to the well-known notion of \emph{frustration index} \cite{Ha59,AR58}. The \emph{frustration index} of a matrix $M$ is the minimum number of non-zero off-diagonal entries of $M$ whose deletion results in a balanced signed graph. A simple reduction from \textsc{MaxCut} shows that computing the frustration index of a signed graph is NP-hard \cite{huffner07}.}.
Our main motivation is to study systems of linear equations in signed matrices that might be ill-defined, and thus do not have a (unique) solution and to minimally modify such matrices so that the resulting linear system becomes (uniquely) solvable. 

\subsection{Results}
Intriguingly, the complexity of \textsc{BoundedEvalueSigning} has not been studied in the literature even though it is widely believed to be a difficult problem in the graph sparsification community. 
We shed light on this problem by showing that it is NP-complete. 
Owing to the close connection between the maximum eigenvalue, positive semi-definiteness, and singularity (by suitable translations), we obtain that \textsc{PsdSigning} and \textsc{SingularSigning} are also NP-complete.


\begin{theorem}\label{thm:bounded-ev}
\textsc{BoundedEvalueSigning}, \textsc{PsdSigning}, and \textsc{SingularSigning} are NP-complete.
\end{theorem}
We remark that the hard instances generated by our proof of Theorem \ref{thm:bounded-ev} are real symmetric matrices with non-zero diagonal entries and hence, it does not completely resolve the computational complexity of the problem of finding a signing of a given \emph{graph-related} matrix (e.g., the adjacency matrix) that minimizes its largest eigenvalue. 




\vspace{9pt}
Our next result provides some evidence that one might be able to design efficient algorithms to solve the NP-complete problems appearing in Theorem \ref{thm:bounded-ev} for \emph{graph-related} matrices. In particular, we show that \textsc{SingularSigning} and its search variant admit efficient algorithms when the input matrix corresponds to the adjacency matrix of a \emph{bipartite} graph. 

\begin{theorem}\label{thm:singular-singing-for-bipartite-graphs}
There exists a polynomial-time algorithm to verify if the adjacency matrix $A_G$ of a given bipartite graph $G$ has a symmetric signing $s$ such that $A_G(s)$ is singular and if so, find such a signing.
\end{theorem}

Next, we show a stark difference in complexity between \textsc{SingularSigning} and \textsc{InvertibleSigning}. 
In contrast to \textsc{SingularSigning} which is NP-complete for arbitrary input matrices (Theorem \ref{thm:bounded-ev}), we show that \textsc{InvertibleSigning} is solvable in polynomial time for arbitrary input matrices. 

\begin{theorem}\label{thm:invertible}
There exists a polynomial-time algorithm to solve \textsc{InvertibleSigning}. 
\end{theorem}

Our algorithm for solving \textsc{InvertibleSigning} is based on a novel graph-theoretic characterization of symmetric matrices $M$ for which every symmetric signed matrix $M(s)$ is singular. We believe that our characterization might be of independent interest. We describe the characterization now.  

The \EMPH{support graph} of a real symmetric $n\times n$ matrix $M$ is an undirected graph $G$ where the vertex set of $G$ is $\Set{1,\dots,n}$, and the edge set of $G$ is $\Set{\bigl. \{u,v\} \mid M[u,v]\neq 0}$. We note that $G$ could have self-loops depending on the diagonal entries of $M$. 
A \EMPH{perfect 2-matching} in a graph $G$ with edge set $E$ is an assignment $x:E\to \{0,1,2\}$ of values to the edges such that $\sum_{e\in \delta(v)}x_e=2$ holds for every vertex $v$ in $G$ (where $\delta(v)$ denotes the set of edges incident to $v$). 
We show the following characterization from which Theorem \ref{thm:invertible} follows immediately.

\begin{theorem}\label{thm:characterization}
Let $M$ be a symmetric $n \times n$ matrix and let $G$ be the support graph of $M$. 
The following are equivalent: 
\begin{enumerate}\itemsep3pt
\item The signed matrix $M(s)$ is singular for every symmetric signing $s$.
\item The support graph $G$ does not contain a perfect 2-matching.
\end{enumerate}
Moreover, there exists a polynomial-time algorithm to verify whether the signed matrix $M(s)$ is singular for every symmetric signing $s$.
\end{theorem}


\paragraph{Remark.} For a subset $S$ of vertices in a graph $G$, let $N_G(S)$ be the \emph{non-inclusive neighborhood} of $S$, that is, $\Set{\bigl. u\in V\setminus S \mid \text{$\Set{u,v}$ is an edge of $G$ for some $v$ in $S$} }$.
A subset $S$ of vertices is said to be \emph{independent} if there are no edges between any pair of vertices in $S$. 
A subset $S$ of vertices is said to be \emph{expanding} in $G$ if $|N_G(S)|\ge |S|$.
Tutte \cite{t-1og-53} showed that the existence of a \emph{non-expanding independent set} is equivalent to the absence of perfect 2-matchings in the graph, which in turn has been used in the study of independent sets \cite{BNR96,al-istgp-07,To14}.
%
%
Thus, Theorem \ref{thm:characterization} can be interpreted as a \emph{spectral characterization} of graphs with non-expanding independent sets: a graph contains a non-expanding independent set if and only if every symmetric signed adjacency matrix of the graph is singular. 

\vspace{9pt}
Our next result focuses on the search variant of \textsc{InvertibleSigning}. We mention that our proof of Theorem \ref{thm:characterization} is non-constructive, that is, even if the support graph of the given matrix has a perfect 2-matching, our proof does not lead to an efficient algorithm to find an invertible signing. While we do not have an efficient algorithm for the search problem for arbitrary symmetric matrices, we obtain an efficient algorithm for those whose support graph is bipartite. 

\begin{theorem}
\label{thm:non-singular-signing-search}
There exists a polynomial-time algorithm to verify if a given symmetric matrix $M$, whose support graph is bipartite, has a symmetric signing $s$ such that $M(s)$ is invertible and if so, find such a signing.
\end{theorem}


We next investigate the number of invertible signed adjacency matrices of graphs with perfect $2$-matchings. In light of Theorem \ref{thm:characterization}, we recall the variant of the randomized algorithm to verify the existence of a perfect matching in a bipartite graph mentioned in Section \ref{sec:motivations}: substitute uniformly random $\{\pm 1\}$ values into variables $x_{e}$ in the Edmonds matrix of $G$ and evaluate the determinant over the reals.  Given a graph $G$ with $n$ vertices and $m$ edges, it can be shown that if $G$ has a perfect matching then the determinant will be nonzero with probability at least $2^{-m}$ using our characterization in Theorem \ref{thm:characterization}. 

We show a much larger lower bound on this probability by obtaining a lower bound on the number of invertible signed adjacency matrices of graphs.
It is well-known that flipping the signs on the edges of a cut preserves the spectrum of the signed adjacency matrix. Thus, the existence of one invertible signed adjacency matrix for a (connected) graph $G$ also implies the existence of $2^{n-1}$ invertible signed adjacency matrices. 
In comparison, the lower bound obtained in our next result is much larger.
We emphasize that our lower bound holds for general simple graphs (and not just bipartite graphs).

\begin{theorem}
\label{thm:non-singular-counting}
Let $G$ be a {simple} graph with $n$ vertices and $m$ edges that has at least one perfect 2-matching.
Then, the number of invertible signed adjacency matrices of $G$ is at least ${2^{m-n}}/{n!}$.
\end{theorem}

Theorem \ref{thm:non-singular-counting} shows that the fraction of invertible signed adjacency matrices of a graph $G$ containing a perfect $2$-matching is at least $2^{-O(n\log n)}$. An upper bound of $2^{-\Omega(n)}$ on the fraction is demonstrated by the graph that is a disjoint union of $4$-cycles.

\vspace{9pt}
Finally, akin to \emph{frustration index}, 
we define the \emph{solvability index} of a real symmetric matrix $M$ to be the smallest number of non-diagonal zero entries that need to be converted to non-zeroes so that the resulting \emph{symmetric} matrix has an invertible symmetric signing. We give an efficient algorithm to find the solvability index of a given symmetric matrix $M$. We emphasize that the support-increase operation that we consider preserves symmetry, that is, if we replace the zero entry $A[i,j]$ by $\alpha$, then the zero entry $A[j,i]$ is also replaced by $\alpha$. 

\begin{theorem}
\label{thm:algo-edge-add}
There exists a polynomial-time algorithm to find the solvability index of a given real symmetric matrix. 
\end{theorem}


\subsection{Organization}
\label{subsec:organize}
In Section \ref{S:prelim}, we review definitions and notations. 
In Section \ref{S:Invertible}, we focus on our results related to invertible signings.  
This includes the combinatorial characterization of matrices with invertible signings (Theorem~\ref{thm:characterization}), an efficient algorithm to find an invertible signing of matrices with bipartite support graphs (Theorem~\ref{thm:non-singular-signing-search}), a lower bound on the number of invertible signed adjacency matrices (Theorem~\ref{thm:non-singular-counting}), and an efficient algorithm to find the solvability index of symmetric matrices (Theorem~\ref{thm:algo-edge-add}). 
In Section \ref{S:Singular}, we focus on our results related to singular signings.  This includes an efficient algorithm to find a singular signing of adjacency matrices of bipartite graphs (Theorem~\ref{thm:singular-singing-for-bipartite-graphs}) and a proof of NP-completeness of \textsc{SingularSigning} (Lemma~\ref{L:singular-signing-npc}).  
In Section \ref{S:hardness-evalue}, we complete Theorem \ref{thm:bounded-ev} by showing that \textsc{PsdSigning} and \textsc{BoundedEvalueSigning} are also NP-complete. 
We conclude by discussing open questions and potential avenues for future research in Section \ref{S:conclusion}.

\subsection{Preliminaries} \label{S:prelim}
Unless otherwise specified, all matrices are symmetric and take values over the reals.
Since all of our results are for symmetric signings, we will just use the term \emph{signing} to refer to a symmetric signing in the rest of this work.  
We now introduce some definitions that will be used in multiple sections.  

A \emph{matching} is a vertex-disjoint union of edges. 
Let $S_n$ denote the set of permutations of $n$ elements. Then, the \emph{permutation expansion} of the determinant of a signed matrix $M(s)$ is given by  
\[
\det M(s) = \sum_{\sigma\in S_n} \sgn(\sigma) \cdot \prod_{i=1}^n M(s)[i,\sigma(i)].
\]
For ease of presentation, let us define $M_\sigma(s) \coloneqq \sgn(\sigma) \cdot \prod_{i} M(s)[i,\sigma(i)]$ and $M_{\sigma} \coloneqq M_{\sigma}(J)$, where $J$ is the signing corresponding to all entries being $+1$.
A permutation $\sigma$ in $S_n$ 
has a unique cycle decomposition and hence 
corresponds to a vertex disjoint union of directed cycles and self-loops on $n$ vertices. Removing the orientation gives an undirected graph which is a vertex disjoint union of cycles of length at least three, matching edges, and self-loops. 

\vspace{9pt}
We use the notion of \emph{Schur complement} repeatedly. The following lemma summarizes the definition and the relevant properties of the Schur complement. 
\begin{lemma}[Horn and Johnson \cite{horn}]
\label{lem:schur-complement-properties}
Suppose $A\in \R^{p\times p}$, $B\in \R^{p\times q}$, $C\in \R^{q\times q}$ are matrices such that $A$ is invertible and the matrix 
\[
D :=
\begin{bmatrix}
A &B \\
B^T & C\\
\end{bmatrix}
\]
is a symmetric matrix.  
Then the \emph{Schur complement} of $C$ in matrix $D$ is defined to be 
\[
D_C:=C - BA^{-1}B^T.
\]
We have the following properties:
\begin{enumerate}
\item[$(i)$] Suppose $A$ is positive definite. Then, $D$ is positive semi-definite if and only if the Schur complement of $C$ in $D$, namely $D_C$, is positive semi-definite. 
\item[$(ii)$] $\det (D) = \det (A) \cdot \det (D_C)$.
\end{enumerate}
\end{lemma}

\section{Invertible Matrices} \label{S:Invertible}

In this section, we focus on invertible signings and prove Theorems \ref{thm:characterization},
\ref{thm:non-singular-signing-search},
\ref{thm:non-singular-counting}, 
and \ref{thm:algo-edge-add}.



\subsection{Matrices with Invertible Signings}
\label{subsection:characterization}




\def\dist{\operatorname{\mathit{dist}}}

In this section, we present a characterization of symmetric matrices $M$ for which every symmetric signed matrix $M(s)$ is singular by proving Theorem \ref{thm:characterization}. 
We start by defining the notation. 

We recall that a permutation $\sigma$ in $S_n$ 
has a unique cycle decomposition and 
corresponds to a vertex disjoint union of directed cycles and self-loops on $n$ vertices. Moreover, removing the orientation gives an undirected graph which is a vertex disjoint union of cycles of length at least three, matching edges, and self-loops. Let the collection of (undirected) edges in the cycle components, matching components, and self-loop components in the resulting undirected graph be denoted by $\cycles(\sigma)$, $\matchings(\sigma)$, and $\loops(\sigma)$ respectively. 
We observe that $\sgn(\sigma)$ is the \emph{parity} of the sum of the number of matching edges and the number of even-length cycles (cycles with even number of edges) in the undirected subgraph induced by the edges in $\cycles(\sigma)\cup \matchings(\sigma)$.  
%
For a matrix $M$ and a signing $s$, we define 
\begin{align*}
M_{\cycles}(\sigma,s) &\coloneqq \Paren{ \prod_{\Set{u,v} \in \cycles(\sigma)}  M(s)[u,v]}, \\
M_{\matchings}(\sigma,s) &\coloneqq \Paren{ \prod_{\Set{u,v} \in \matchings(\sigma)} M(s)[u,v]^2 } \text{, and}\\
M_{\loops}(\sigma,s) &\coloneqq \Paren{ \prod_{\Set{u,u} \in \loops(\sigma)} M(s)[u,u] }.
\end{align*}
We use the convention that a product over an empty set is equal to $1$.
With this notation, we have
\[
M_\sigma(s) = \sgn(\sigma) \cdot M_{\cycles}(\sigma,s) \cdot M_{\matchings}(\sigma,s) \cdot M_{\loops}(\sigma,s).
\]

The following lemma is the core of our proof of Theorem \ref{thm:characterization}.
\begin{lemma}
\label{L:zero-terms}
Let $M$ be a symmetric $n \times n$ matrix. Then $M_\sigma=0$ holds for every permutation $\sigma$ in $S_n$ if and only if the signed matrix $M(s)$ is singular for all signing $s$.
\end{lemma}

\begin{proof}
The forward implication follows immediately: If $M_{\sigma}=0$ for every permutation $\sigma$ in $S_n$, then for every signing $s$, we have $M_{\sigma}(s)=0$ for every permutation $\sigma$ in $S_n$. Hence, every term in the permutation expansion of the determinant of $M(s)$ is zero for every signing $s$.

We now show the reverse implication. 
By assumption, $M(s)$ is singular for all signings $s$:
\[
\det(M(s)) = \sum_{\sigma\in S_n} M_\sigma(s) = 0.
\]

Assume for contradiction that there is a permutation $\tau$ such that $M_\tau\neq 0$.
Let $\Gamma$ be a \emph{minimum cardinality} subset of permutations in $S_n$ such that 
\begin{enumerate}\itemsep=2pt
\item[($i$)] $\tau\in \Gamma$ and
\item[($ii$)] $\sum_{\sigma\in\Gamma} M_\sigma(s)=0$ for all signings $s$.
\end{enumerate}
We observe that such a set $\Gamma$ exists since $S_n$ is a valid choice for $\Gamma$. 
The following claim, which we prove later, shows that the cycles and the self-loops of all permutations in $\Gamma$ coincide with that of $\tau$. 
\begin{claim}\label{claim:coinciding-cycles-loops}
For every permutation $\sigma$ in $\Gamma$, we have $\cycles(\sigma) = \cycles(\tau)$ and $\loops(\sigma) = \loops(\tau)$.
\end{claim}

Now we show that Lemma \ref{L:zero-terms} follows from the claim.  Let us fix a permutation $\sigma$ in $\Gamma$.  
By Claim \ref{claim:coinciding-cycles-loops}, we have $\cycles(\sigma) = \cycles(\tau)$ and $\loops(\sigma) = \loops(\tau)$---i.e., the set of edges in cycle components and the set of edges in self-loop components in $\sigma$ and $\tau$ coincide.  Furthermore, Claim \ref{claim:coinciding-cycles-loops} also implies that the number of matching edges in $\sigma$ and $\tau$ is the same and hence, $\sgn(\sigma) = \sgn(\tau)$.  
So we have
\begin{align*}
M_\sigma(s) 
&= \sgn(\sigma) \cdot M_{\cycles}(\sigma,s) \cdot M_{\matchings}(\sigma,s) \cdot M_{\loops}(\sigma,s) \\
&= \sgn(\tau) \cdot M_{\cycles}(\tau,s) \cdot M_{\matchings}(\sigma,s) \cdot M_{\loops}(\tau,s) \\
&= M_\tau(s) \cdot \Paren{\frac{M_{\matchings}(\sigma,s)}{M_{\matchings}(\tau,s)}}.
\end{align*}
%
%
Using the above expression and taking the sum of the terms $M_\sigma(s)$ over all $\sigma$ in $\Gamma$, we have 
\[
\sum_{\sigma \in \Gamma} M_\sigma(s) = \frac{M_\tau(s)}{M_{\matchings}(\tau,s)} \cdot \Paren{ \sum_{\sigma \in \Gamma} M_{\matchings}(\sigma,s)}.
\]
By the choice of $\tau$, we know that $M_\tau(s)/M_{\matchings}(\tau,s)$ is non-zero. Moreover, $M_{\matchings}(\sigma,s)$ is a perfect square and therefore 
non-negative
for all permutations $\sigma$ in $\Gamma$. In particular, since $\tau$ is in $\Gamma$, we have that $\Gamma$ is non-empty and hence, 
\[
\sum_{\sigma \in \Gamma} M_{\matchings}(\sigma,s)\neq 0. 
\] 
Consequently, 
the sum $\sum_{\sigma \in \Gamma} M_\sigma(s)$ is non-zero, contradicting condition $(ii)$ in the choice of $\Gamma$.
\end{proof}

We now prove Claim \ref{claim:coinciding-cycles-loops}. 
\begin{proof}[Proof of Claim \ref{claim:coinciding-cycles-loops}]
Let us consider an arbitrary permutation $\tau'\in \Gamma\setminus \{\tau\}$. 
%
Assume for contradiction that there is an edge $e$ in the symmetric difference of $\cycles(\tau) \cup \loops(\tau)$ and $\cycles(\tau') \cup \loops(\tau')$.
Partition $\Gamma$ into two subsets $\Gamma_e$ and $\Gamma_{e}'$, where $\Gamma_e$ contains any permutation $\sigma$ in $\Gamma$ that has the edge $e$ in the subgraph induced by the edges in $\cycles(\sigma) \cup \loops(\sigma)$, and $\Gamma_{e}':=\Gamma\setminus \Gamma_e$. 
By this partitioning, for every permutation $\sigma'\in \Gamma_e'$, either $e\in \matchings(\sigma')$ or $e\not\in \cycles(\sigma')\cup\matchings(\sigma')\cup\loops(\sigma')$ holds. 
We also observe that exactly one of the permutations $\tau$ and $\tau'$ is in $\Gamma_e$ while the other is in $\Gamma_{e}'$.  In particular, both $\Gamma_e$ and $\Gamma_e'$ are nonempty. We will show that either $\Gamma_e$ or $\Gamma_e'$ contradicts the minimality of $\Gamma$. 
%

Let us consider an arbitrary signing $s$, and let $s'$ be another signing that is 
obtained from $s$ by flipping the sign only on edge $e$. 
Consider the following four sums:
\begin{align*}
\Sigma_{00} \coloneqq \sum_{\sigma\in{\Gamma_e}} M_\sigma(s), \quad 
\Sigma_{01} \coloneqq \sum_{\sigma\in{\Gamma_e}} M_\sigma(s'), \quad 
\Sigma_{10} \coloneqq \sum_{\sigma\in\Gamma_{e}'} M_\sigma(s), \quad
\Sigma_{11} \coloneqq \sum_{\sigma\in\Gamma_{e}'} M_\sigma(s').
\end{align*}

Now, by condition $(ii)$ in the choice of $\Gamma$, we have that 
\begin{align}
\Sigma_{00}+\Sigma_{10} =0 \label{eqn:sum1},\\
\Sigma_{01}+\Sigma_{11} =0 \label{eqn:sum2}.
\end{align}

For every $\sigma$ in $\Gamma_e$, we have that $M_{\sigma}(s)=-M_{\sigma}(s')$ since the edge $e=\Set{u,v}$ is in the subgraph induced by the edges in $\cycles(\sigma)\cup\loops(\sigma)$ and hence exactly one of the two terms $M(s)[u,v]$ and $M(s)[v,u]$ appears in $M_{\sigma}(s)$. Therefore, 
\begin{align}
\Sigma_{00} = -\Sigma_{01}. \label{eqn:negative}
\end{align}

For every $\sigma'$ in $\Gamma_e'$, we have that $M_{\sigma'}(s)=M_{\sigma'}(s')$ since we have either $e\in \matchings(\sigma')$ or $e\not\in \matchings(\sigma')\cup\cycles(\sigma')\cup\loops(\sigma')$ holds and in both cases, an even number of terms among $M(s)[u,v]$ and $M(s)[v,u]$ appear in $M_{\sigma'}(s)$. Therefore, 
\begin{align}
\Sigma_{10} = \Sigma_{11}. \label{eqn:positive}
\end{align}
By equations (\ref{eqn:sum1}), (\ref{eqn:sum2}), (\ref{eqn:negative}), and (\ref{eqn:positive}), we have that $\Sigma_{00}=\Sigma_{01}=\Sigma_{10}=\Sigma_{11}=0$ for every signing $s$. 

Now, let us take $T$ to be the set in $\{\Gamma_e,\Gamma_e'\}$ that contains $\tau$. We obtain that ($i$) $\tau\in T$ and ($ii$) $\sum_{\sigma\in T}M_\sigma(s)=0$ for every signing $s$. Moreover $|T|<|\Gamma|$, contradicting the minimality of $\Gamma$. 
\end{proof}
\paragraph{Remark.} 
Lemma \ref{L:zero-terms} can also be shown using the DeMillo-Lipton-Schwartz-Zippel lemma \cite{s-fpavp-80,z-pasp-79,dl-prapt-78} by exploiting the low-degree nature of the multivariate determinant polynomial. This alternative proof, upon defining the notations, also takes the same amount of space as the above proof. We avoid presenting multiple proofs of the same result in the interests of brevity. 

\vspace{9pt}
To complete the proof of Theorem~\ref{thm:characterization}, we use the following lemma about the complexity of verifying the existence of a perfect $2$-matching in a given graph. The lemma follows from a well-known reduction to the perfect matching problem in bipartite graphs. 

\begin{lemma}[Tutte \cite{t-1og-53}; Lov\'{a}sz and Plummer~{\cite[Corollary~6.1.5]{lp-mt-09}}]
\label{L:algo_find_perfect}
There exists a polynomial-time algorithm to verify if a given graph (possibly with loops) has a perfect $2$-matching. 
\end{lemma}

We now prove Theorem \ref{thm:characterization}. 

%
\begin{proof}[Proof of Theorem~\ref{thm:characterization}]
By Lemma~\ref{L:zero-terms}, the signed matrix $M(s)$ is singular for every signing $s$ if and only if $M_\sigma=0$ holds for every permutation $\sigma$ in $S_n$. 
The existence of a perfect 2-matching in the support graph of $M$ is equivalent to the fact that $M_\sigma \neq 0$ for some $\sigma$ in $S_n$, and therefore
we have that $M_\sigma=0$ for every $\sigma$ in $S_n$ if and only if the support graph of $M$ has no perfect $2$-matchings. 
Moreover, Lemma \ref{L:algo_find_perfect} immediately gives us a polynomial-time algorithm to verify whether the signed matrix $M(s)$ is singular for every signing $s$.
\end{proof}

\subsection{Finding Invertible Signings of Bipartite Graphs} \label{SS:bipartite-invertible-algorithm}
In this section, we present an algorithm to find an invertible signing of the adjacency matrix of a given bipartite graph and thus prove Theorem \ref{thm:non-singular-signing-search}.  

%
%

We say that a signing $s'$ \EMPH{extends} another signing $s$ \emph{on entry} $(u,v)$ if $s'[i,j] = s[i,j]$ for every entry $(i,j) \not\in \{(u,v),(v,u)\}$. Thus, if $s'$ extends a signing $s$ on entry $(u,v)$, then $s'$ could be $s$ or it differs from $s$ only in the entry corresponding to $u$'th row and $v$'th column (and by symmetry, the entry corresponding to $v$'th row and $u$'th column).
We now have the ingredients to show that incrementing a signing while preserving invertibility is possible.

\begin{lemma}[Incremental Signing]\label{lem:partial-signing-bipartite}
Let $G$ be a bipartite graph 
with bipartition $(L,R)$ of the vertex set, and let $A_G$ be the adjacency matrix of $G$.  Suppose there exists a signing $s$ such that $A_G(s)$ is invertible. Let $\ell\in L, r\in R$ be vertices in $G$ such that $e \coloneqq \{\ell,r\}$ is not an edge of $G$. Then there exists a signing $s'$ that extends $s$ on $(\ell,r)$ such that $A_{G+e}(s')$ is invertible, where $G+e$ is the graph obtained by adding the edge $e$ to $G$.
\end{lemma}

\begin{proof}
Let $n$ be the number of vertices in $G$. 
Let $s$ be a signing such that $A_G(s)$ is invertible. 
Let $s'$ be an $n \times n$ matrix where $s'[i,j] = s[i,j]$ for all pairs of $(i,j)$ besides $(\ell,r)$ and $(r,\ell)$, and set $s'[\ell,r]$ (and thus by symmetry, $s'[r,\ell]$) to be a variable $x$. Let $b_\ell$ and $b_r$ be vectors of length $n-2$ such that $b_\ell[i]  = A_G(s')[\ell,i]$ and $b_r[i] = A_G(s')[r,i]$ for every $i$ not equal to $\ell$ or 
$r$. For a subgraph $G'$ of $G$ with adjacency matrix $A_{G'}$, let $A_{G'}(s)$ denote the signed adjacency matrix of $G'$ obtained by the entry-wise product of $A_{G'}$ and the signing obtained by projecting $s$ to the edges of $G'$.

Consider the matrix $A_{G+e}(s')$ obtained by taking entry-wise product of $A_{G+e}$ and $s'$. Let $H$ be the graph obtained by removing vertices $r$ and $\ell$ from $G$, and let $A_H$ be the adjacency matrix of $H$. 
In the notation defined, we have
\begin{align*}
A_{G+e}(s') = 
\begin{bmatrix}
	0 & x & \multicolumn{2}{c}{b_\ell} \\
	x & 0 & \multicolumn{2}{c}{b_r} \\
	b_\ell^T & b_r^T &\multicolumn{2}{c}{A_H(s)} \\
\end{bmatrix}
\end{align*}
with the first and second rows (by symmetry, columns) corresponding to vertices $\ell$ and $r$ respectively. Let $f(x):=\det(A_{G+e}(s'))$. We have that 
\begin{align*}
f(x) 
&= - \det \left(
	A_H(s) 
\right)x^2 
- \det \left(
\begin{bmatrix}
	0 & \multicolumn{2}{c}{b_r} \\
	b_\ell^T&\multicolumn{2}{c}{A_H(s)} \\
\end{bmatrix}
\right)x 
+
\det \left(
A_G(s)
\right).
\end{align*}
We note that $f(x)$ is a quadratic function of $x$. Now suppose for the sake of contradiction that the matrix $A_{G+e}(s')$ is singular for both $x = \pm1$. Then $f(1)=f(-1)=0$ and hence the following holds.
\begin{align}
&\det \left(
\begin{bmatrix}
	0 & \multicolumn{2}{c}{b_r} \\
	b_\ell^T&\multicolumn{2}{c}{A_H(s)} \\
\end{bmatrix}
\right) = 0 \label{eqn:bia-1}\\
&\det(A_H(s)) =  \det(A_G(s)) \label{eqn:bia-2}
\end{align}
We recall that $\det(A_G(s))\neq 0$ and hence $\det(A_H(s))\neq 0$ by equation~(\ref{eqn:bia-2}). Since $\det(A_H) \not = 0$, we use the property of the Schur complement (Lemma \ref{lem:schur-complement-properties}) to obtain that
\begin{align*}
\det(A_G(s)) &= 
\det(A_H(s)) \cdot
\det
\left(
0 - 
\begin{bmatrix}
	b_\ell \\
	b_r \\
\end{bmatrix}
A_H(s)^{-1}
\begin{bmatrix}
	b_\ell^T & b_r^T
\end{bmatrix}
\right)\\
&= \det(A_H(s)) \cdot 
\det \left(
\begin{bmatrix}
	b_\ell A_H(s)^{-1} b_\ell^T & b_\ell A_H(s)^{-1} b_r^T \\
	b_r A_H(s)^{-1} b_\ell^T &  b_r A_H(s)^{-1} b_r^T \\
\end{bmatrix}
\right).
\end{align*}
Using equation (\ref{eqn:bia-2}), we thus have 
\begin{equation}\label{func:func1}
\det \left(
\begin{bmatrix}
	b_\ell A_H(s)^{-1 }b_\ell^T & b_\ell A_H(s)^{-1} b_r^T \\
	b_r A_H(s)^{-1} b_\ell^T &  b_r A_H(s)^{-1} b_r^T \\
\end{bmatrix}
\right) = 1.
\end{equation}
Let $G - \ell$ and $G - r$ be the graphs obtained by removing vertices $\ell$ and $r$ from $G$ respectively. 
Then by applying the Schur complement on $A_{G-r}(s)$ (Lemma~\ref{lem:schur-complement-properties}), we have that 
\begin{align}
\det (A_{G-r}(s)) 
= 
\det \left(
\begin{bmatrix}
	0 & b_r \\
	b_r^T & A_H(s)\\
\end{bmatrix}
\right) 
&=\det(A_H(s)) \cdot \det(0 - b_r A_H(s)^{-1} b_r^T),
\end{align}
and hence
\begin{align}
\det (A_{G-r}(s)) &= -\det(A_H(s)) \cdot b_r A_H(s)^{-1} b_r^T.\label{eqn:bia-3}
\end{align}
Similarly, we also have 
\begin{align}
\det (A_{G-\ell}(s)) 
&= 
 -\det(A_H(s)) \cdot b_\ell A_H(s)^{-1} b_\ell^T. \label{eqn:bia-4}
\end{align}
Moreover, by equation (\ref{eqn:bia-1}) and the property of Schur complement (Lemma~\ref{lem:schur-complement-properties}), we have that 
\begin{align}
0=
\det \left(
\begin{bmatrix}
	0 & b_r \\
	b_\ell^T & A_H(s)\\
\end{bmatrix}
\right) 
&=
\det(A_H(s)) \cdot \det(0 - b_r A_H(s)^{-1} b_\ell^T). \notag
\end{align}
Hence,
\begin{align}
b_r A_H(s)^{-1} b_\ell^T&=0. \label{eqn:bia-5}
\end{align}
Similarly, we also have 
\begin{align}
b_{\ell} A_H(s)^{-1} b_r^T&=0. \label{eqn:bia-6}
\end{align}
Thus, using equations (\ref{eqn:bia-3}), (\ref{eqn:bia-4}), (\ref{eqn:bia-5}), and (\ref{eqn:bia-6}), we have 
 \begin{align}
\det \left(
\begin{bmatrix}
	b_\ell A_H(s)^{-1} b_\ell^T & b_\ell A_H(s)^{-1} b_r^T \\
	b_r A_H(s)^{-1} b_\ell^T &  b_r A_H(s)^{-1} b_r^T \\
\end{bmatrix}
\right) 
&=
\frac{\det(A_{G-r}(s))}{\det(A_H(s))} \cdot \frac{\det(A_{G-\ell}(s))}{\det(A_H(s))}.\label{eqn:bia-7}
\end{align}
However, since $G$ is bipartite and has a perfect 2-matching, the subgraphs $G-r$ and $G-\ell$ must  be bipartite and have an odd number of vertices. Hence, the subgraphs $G-r$ and $G-\ell$ have no perfect 2-matching. Thus, by Lemma~\ref{L:algo_find_perfect} and the backward direction of Theorem~\ref{thm:characterization},  we have $\det(A_{G-r}(s)) = \det(A_{G-\ell}(s)) = 0$ which together with equation (\ref{eqn:bia-7}) contradicts equation~(\ref{func:func1}).
\end{proof}

\begin{figure}[ht]
\centering
\small
\begin{algorithm}
	\label{algo:find_bipartite_invertible_signing}
	\textul{$\mathsc{FindInvertibleSigningBipartite}(G)$:} \quad {\textit{Input}: A bipartite graph $G$.} \+ 
\\	1. Find a perfect matching $M$ in $G$.
\\	2. Let $H$ be the subgraph of $G$ with edge set $M$.
\\	3. Let $s$ be the all-one signing. 
\\	4. While $H \not = G$: \+
\\		4.1. Let $e:=\{\ell,r\}$ be an edge in $G$ but not in $H$.
\\		4.2. Find a signing $s'$ that extends $s$ on $(\ell,r)$ such that $A_{H + e}(s')$ is invertible.
\\		4.3. Update $s\leftarrow s'$ and $H\leftarrow H+e$. \-
\\	5. Return $s$.
\end{algorithm}
\caption{The algorithm $\mathsc{FindInvertibleSigningBipartite}(G)$.}
\label{F:find_bipartite_invertible_signing}
\end{figure}

Lemma \ref{lem:partial-signing-bipartite} suggests a natural algorithm to find an invertible signing of the adjacency matrix of a given bipartite graph in polynomial time that is presented in Figure~\ref{F:find_bipartite_invertible_signing}. The correctness of the algorithm follows from Lemma \ref{lem:partial-signing-bipartite}. It can be implemented to run in polynomial time since a perfect matching in a bipartite graph can be found efficiently and moreover, Step 4.2 only requires us to consider the determinant of the signed adjacency matrix of $H+e$ for the two possible 
signings $s'$ that extend $s$ on $(\ell,r)$ 
(where the two extensions are obtained by signing the edge $e$ as $\pm1$). This completes the proof of Theorem \ref{thm:non-singular-signing-search}.
Our algorithm also gives an alternative constructive proof of Theorem \ref{thm:characterization} for matrices whose support graph is bipartite.


\subsection{Number of Invertible Signings}
\label{SS:counting}

In this section, we show a lower bound on the number of invertible signings of adjacency matrices of graphs with at least one perfect $2$-matching, thus proving Theorem~\ref{thm:non-singular-counting}.

\def\supp{\mathrm{supp}}

\begin{proof}[Proof of Theorem~\ref{thm:non-singular-counting}]
Let $A$ denote the adjacency matrix of $G$ while $n$ and $m$ be the number of vertices and edges in $G$ respectively.
For a vector $s\in \{\pm 1\}^{E(G)}$ that gives a sign for the edges of $G$, we define $A(s)$ to be the signed adjacency matrix of $G$ obtained by taking the entry-wise product of $A$ and 
the $n\times n$ symmetric signing whose entry on $(i,j)$ (and by symmetry, $(j,i)$) is $s(\{i,j\})$ if $\{i,j\}$ is an edge in $G$, and $+1$ otherwise.  
We will abuse notation for the purposes of this proof only and call such a vector $s$ to be a \emph{signing} of the graph $G$.
Let 
\begin{align*}
\Gamma_{1}&:=\Set{s \in \{\pm 1\}^{E(G)}\mid \text{$A(s)$ is invertible}} \text{, and}\\
\Gamma_{0}&:=\Set{s \in \{\pm 1\}^{E(G)}\mid \text{$A(s)$ is singular}}.
\end{align*}
For a signing $s$ of $G$ in $\Gamma_0$, we have
\[
\sum_{\omega \in S_n} A_\omega(s) = 0.
\]
Assume for the sake of contradiction that the set $\Gamma_{1}$ has size less than $2^{m-n}/n!$.


By assumption there is a perfect $2$-matching in $G$. Hence there exists a permutation $\tau\in S_n$ for which $A_{\tau}\neq 0$ and hence $A_{\tau}\in \{\pm 1\}$. Fix a permutation $\tau\in S_n$ with the least number of cycles and loops for which $A_{\tau}\neq 0$ (i.e., the perfect $2$-matching corresponding to $\tau$ having the largest number of matching edges).  
Let 
\[
\Omega_0:=\Set{\sigma\in S_n:\cycles(\sigma)=\cycles(\tau)\text{ and } \loops(\sigma)=\loops(\tau)}.
\]
Then for every $\sigma$ in $\Omega_0$ and every signing $s$, we have $A_\sigma(s) = A_\tau(s)$. Now let 
\[
\mathcal{R}:=\Set{s \in \{\pm 1\}^{E(G)} \mid \text{$s(e)=+1$ for every edge $e$ in $\cycles(\tau)$ and $\loops(\tau)$}}. 
\]
Consider the absolute value of double-sum
\[
\left|\sum_{s \in \mathcal{R}} {\sum_{\omega \in S_n} A_\omega(s)}\right|.
\]
We have the upper bound
\[
\left|\sum_{s \in \mathcal{R}} {\sum_{\omega \in S_n} A_\omega(s)}\right|
= \left|\sum_{s \in \mathcal{R}\cap \Gamma_{1}} {\sum_{\omega \in S_n} A_\omega(s)}\right|
\le \sum_{s \in \mathcal{R}\cap \Gamma_{1}} n!
\le |\Gamma_{1}| \cdot n!
< 2^{m-n}.
\]

Now we will show a lower bound to the double-sum that contradicts the upper bound. We first note that
\[
\sum_{s \in \mathcal{R}} {\sum_{\omega \in S_n} A_\omega(s)}
=\sum_{s \in \mathcal{R}} \sum_{\omega \in \Omega_0} A_\omega(s) + \sum_{s \in \mathcal{R}} \sum_{\omega \in S_n \setminus \Omega_0} A_\omega(s).
\]
Since $\tau$ is chosen to have the fewest number of cycles and loops, it follows that any permutation in $S_n \setminus \Omega_0$ has at least one cycle or loop edge that is not in $\tau$, which implies that the signings in $\mathcal{R}$ can be paired up such that each pair of signings differ only on that specific edge not in $\tau$.  Therefore for each permutation $\omega \in S_n \setminus \Omega_0$, we have 
\[
\sum_{s \in \mathcal{R}} A_\omega(s) = 0.
\]
Hence, 
\[
\sum_{s \in \mathcal{R}} \sum_{\omega \in S_n \setminus \Omega_0} A_\omega(s) = 0.
\]

Now fix any $s_0\in \mathcal{R}$ and we have 
\begin{align*}
\left|\sum_{s \in \mathcal{R}} {\sum_{\omega \in S_n} A_\omega(s)} \right|
= \left|\sum_{s \in \mathcal{R}} {\sum_{\omega \in \Omega_0} A_\omega(s)}\right|
= |\mathcal{R}||\Omega_0| |A_\tau(s_0)|
\ge |\mathcal{R}|.
\end{align*}
The second equation above is because $A_{\omega}(s)=A_{\tau}(s_0)$ for every $\omega\in \Omega_0$ and every $s\in \mathcal{R}$. 
The last inequality above is because $\tau\in \Omega_0$ and hence $|\Omega_0|\ge 1$ and moreover, $A_{\tau}(s_0)\in \{\pm 1\}$. 

Thus the cardinality of $\mathcal{R}$ is a lower bound for the absolute value of the double-sum of interest. Now we note that $\mathcal{R}$ is the set of all signings on edges for which $s(e)=+1$ on every edge $e$ in $\cycles(\tau)\cup \loops(\tau)$. The total number of edges in $\cycles(\tau)\cup \loops(\tau)$ is at most $n$ and hence $|\mathcal{R}|\ge 2^{m-n}$. Thus, we have 
\[
\left|
\sum_{s \in \mathcal{R}} {\sum_{\omega \in S_n} A_\omega(s)} \right|
\ge |\mathcal{R}| \ge 2^{m-n},
\]
a contradiction to the upper bound on the absolute value of the double-sum.
\end{proof}

\paragraph{Remark.} We mention that the above proof of Theorem \ref{thm:non-singular-counting} can also be adapted to prove Theorem \ref{thm:characterization}. We avoid presenting multiple proofs of the same result in the interests of brevity.
\subsection{Minimum Support Increase to Obtain an Invertible Signing}
\label{subsec:getting-perfect-2-matching}

In this section, we study the problem of computing the solvability index of real symmetric matrices, thus proving Theorem \ref{thm:algo-edge-add}. We recall the following definition:
%
For a real symmetric matrix $M$, the \emph{solvability index} of $M$ is the smallest number of non-diagonal zero entries that need to be converted to non-zeroes so that the resulting \emph{symmetric} matrix has an invertible signing. We remind the reader that the support-increase operation preserves symmetry.

By our characterization in Theorem \ref{thm:characterization}, computing the solvability index of a matrix reduces to the following edge addition problem:

\vspace{9pt}
\noindent {\textsc{EdgeAdd}}: Given a graph $G$ (possibly with self-loops) with vertex set $V$ and edge set $E$, find
\[
\min \Set{\bigl. |F| \mid \text{$F$ is a set of non-edges of $G$ with no loops such that $G+F$ has a perfect $2$-matching}}.
\]
In the above, 
$G+F$ denotes the graph obtained by adding the edges in $F$ to $G$. 
In the rest of the section, we will show that \textsc{EdgeAdd} can be solved efficiently, which will imply Theorem \ref{thm:algo-edge-add}. 

\begin{theorem}
\label{algo:edge-add}
There is a polynomial-time algorithm to solve \textsc{EdgeAdd}. 
\end{theorem}

We need some terminology from matching theory. Let $G$ be a graph on vertex set $V$ and edge set $E$.  For a subset $S$ of vertices, denote the \emph{induced subgraph} of $G$ on $S$ as $G[S]$ and the \emph{non-inclusive neighborhood} of $S$ in $G$ by $N_G(S)$.
We recall that a \emph{matching} $M$ in $G$ is a subset of edges where each vertex is incident to at most one edge in $M$. 
Let $\nu(G)$ denote the cardinality of a \emph{maximum matching} in $G$ and let 
\[
\nu_f(G) \coloneqq \max \Set{ \sum_{e\in E}x_e \,\Bigg|\, \text{$\sum_{e\in \delta(v)}x_e\le 1$, and $x_e\ge 0$ for all $e\in E$} }
\] 
denote the value of a \emph{maximum fractional matching} in $G$.  
For a matching $M$, we define a vertex $u$ to be \emph{$M$-exposed} if none of the edges of $M$ are incident to $u$, and a vertex $v$ to be an \emph{$M$-neighbor} of $u$ if edge $\{u,v\}$ is in $M$. 
A vertex $u$ in $V$ is said to be \emph{inessential} if there exists a maximum cardinality matching $M$ in $G$ such that $u$ is $M$-exposed, and is said to be \emph{essential} otherwise. 
A graph $H$ is \emph{factor-critical} if there exists a perfect matching in $H-v$ for every vertex $v$ in $H$. The following result is an immediate consequence of the odd-ear decomposition characterization of Lov\'{a}sz \cite{Lov72}.
\begin{lemma}[Lov\'{a}sz \cite{Lov72}]
\label{lem:fc-graphs-perfect-2-matchings}
If $G$ is a factor-critical graph, then $G$ has a perfect $2$-matching. 
\end{lemma}
The \emph{Gallai-Edmonds decomposition} \cite{gallai63,gallai1964maximale,edmonds1965paths} of a graph $G$ is a partition of the vertex set of $G$ into three sets $(B,C,D)$, where $B$ is the set of inessential vertices, $C:=N_G(B)$, and $D:=V\setminus (B\cup C)$. Let $B_1$ denote the set of isolated vertices in $G[B]$ and $B_{\ge 3}:=B\setminus B_1$. 
For notational convenience, we will denote the Gallai-Edmonds decomposition as $(B=(B_1,B_{\ge 3}), C, D)$. The Gallai-Edmonds decomposition of a graph is unique and can be found efficiently \cite{edmonds1965paths}.
The following theorem summarizes the properties of the Gallai-Edmonds decomposition that we will be using (properties ($i$) and ($ii$) are well-known and can be found in Schrijver \cite{schrijver2003combinatorial} while property ($iii$) follows from results due to Balas \cite{Bal81} and Pulleyblank \cite{pulleyblank87}---see Bock \etal\ \cite{BCKPS14-j} for a proof of property ($iii$)):
\begin{theorem}
\label{thm:ged-properties}
Let $(B=(B_1, B_{\ge 3}),C,D)$ be the Gallai-Edmonds decomposition of a graph $G$.  
%
We have the following properties:
\begin{enumerate}\itemsep=3pt
\item[($i$)] Each connected component in $G[B]$ is factor-critical.  
\item[($ii$)] Every maximum matching $M$ in $G$ contains a perfect matching in $G[D]$ and matches each vertex in $C$ to distinct components in $G[B]$. 
\item[($iii$)] Let $M$ be a maximum matching that matches the largest number of $B_1$ vertices. 
Then there are $2(\nu_f(G)-\nu(G))$ $M$-exposed vertices in $B_{\ge 3}$.  
\end{enumerate}
\end{theorem}

We observe that $G$ contains a perfect $2$-matching if and only if $\nu_f(G)=|V|/2$. Therefore, adding edges to get a perfect $2$-matching in $G$ is equivalent to adding edges to increase the maximum fractional matching value to $|V|/2$. 

\begin{proof}[Proof of Theorem \ref{algo:edge-add}]%
We will assume that $G$ has no isolated vertices and no self-loops in the rest of the proof. We make this assumption here in order to illustrate the main idea underlying the algorithm. This assumption can be relaxed by a case analysis in the algorithm as well as the proof of correctness. We defer the details of the case analysis to the full-version of the paper.

\begin{figure}[ht]
\centering
\small
\begin{algorithm}
	\label{algo1}
	\textul{$\mathsc{EdgeAdd}(G)$:} \quad {\textit{Input}: A graph $G$ with no isolated vertices and no self-loops.} \+ 
\\	1. Find the Gallai-Edmonds decomposition $(B=(B_1, B_{\ge 3}),C,D)$ of $G$.
\\	2. Find a maximum matching $M$ that matches the largest number of $B_1$ vertices.  
\\	3. Let $S \coloneqq \{u\in B_1 \mid \text{$u$ is $M$-exposed}\}$.
\\	4. If $|S|$ is even: \+
\\		Pick an arbitrary pairing of the vertices in $S$. \-
\\	5. If $|S|$ is odd: \+
\\		Consider a vertex $s$ in $S$, pick a vertex $t$ in $N_G(s)$ and let $u$ be the $M$-neighbor of $t$.
\\		Pair up $u$ with $s$ and pick an arbitrary pairing of the vertices in $S\setminus \{s\}$. \-
\\	6. Return the set of pairs $F$.
\end{algorithm}
\caption{The algorithm $\mathsc{EdgeAdd}(G)$.}
\label{F:edge-add}
\end{figure}

We use the algorithm $\mathsc{EdgeAdd}(G)$ given in Figure~\ref{F:edge-add}.
We briefly describe an efficient implementation for Step~2, since it is easy to see that other steps can be implemented efficiently. In order to find a maximum matching that matches the largest number of $B_1$ vertices (as mentioned in property ($iii$) of Theorem \ref{thm:ged-properties}), we first find the Gallai-Edmonds decomposition and a maximum matching $M$. Then, we repeatedly augment $M$ by searching for $M$-alternating paths (of even-length) from $M$-exposed $B_1$ vertices. This approach can be implemented to run in polynomial time. Alternatively, Step~2 can also be implemented by solving a maximum weight matching with suitably chosen weights. 

We now argue the correctness of the algorithm. We first show that if $|S|$ is odd, then there is a choice of vertices $t$ and $u$ as described in the algorithm $\mathsc{EdgeAdd}(G)$: this is because, $G$ has no isolated vertices and hence there exists a vertex $t$ in $N_G(s)$. Moreover, by Theorem \ref{thm:ged-properties}, since $s$ is in $B_1$, it follows that $t$ is in $C$ and thus $t$ is matched by $M$ to a node $u$ in $B$. 
Now, Claim \ref{C:complete} proves feasibility and bounds the size of the returned solution $F$ while Claim \ref{C:sound} proves the optimality.
\end{proof}


\begin{claim}
\label{C:complete}
The algorithm $\mathsc{EdgeAdd}(G)$ returns a set $F$ of non-edges of $G$ such that 
(1) $G+F$ contains a perfect $2$-matching, and
(2) $|F|=\Ceil{\bigl. |V|/2-\nu_f(G) }$.
\end{claim}
\begin{proof}
By property ($ii$) of Theorem \ref{thm:ged-properties}, the set $F$ is a set of non-edges of $G$. We will construct a perfect $2$-matching in $G+F$. By property ($i$) of Theorem \ref{thm:ged-properties}, every component in $G[B_{\ge 3}]$ is factor-critical. By Lemma \ref{lem:fc-graphs-perfect-2-matchings}, every component $K$ in $G[B_{\ge 3}]$ contains a perfect $2$-matching $x^K$. Let $N_K$ denote the support of $x^K$. Let $\mathcal{K}$ denote the components in $G[B_{\ge 3}]$ that contain an $M$-exposed vertex. We have two cases:

\begin{enumerate}
\item[] \textit{Case 1:} Suppose $|S|$ is even. 
Let $N$ denote the set of edges of $M$ that do not match any vertices in $\bigcup_{K\in \mathcal{K}}V(K)$. Now, the set of edges induced by $\Paren{\bigcup_{K\in \mathcal{K}}N_K} \cup N\cup F$ has a perfect $2$-matching.  A perfect $2$-matching $x$ in $G+F$ can be obtained by assigning $x(e):=x^K(e)$ for edges $e$ in $\bigcup_{K\in \mathcal{K}}N_K$, $x(e):=2$ for edges $e$ in $N\cup F$, and $x(e):=0$ for the remaining edges in $G+F$.
\item[] \textit{Case 2:} Suppose $|S|$ is odd. 
Let $N$ denote the set of edges of $M\setminus\{\{t,u\}\}$ that do not match any vertices in $\bigcup_{K\in \mathcal{K}}V(K)$. Now, $\Paren{\bigcup_{K\in \mathcal{K}}N_K} \cup N\cup (F\setminus \{s,u\})\cup\{\{t,u\},\{s,t\},\{s,u\}\}$ has a perfect $2$-matching.  
We note that the edges $\{t,u\}, \{s,t\}$ were already present in the graph owing to the choice of $c$ and $u$ while the edge $\{s,u\}$ was added as an edge from $F$. 
A perfect $2$-matching $x$ in $G+F$ can be obtained by assigning $x(e):=x^K(e)$ for edges $e\in \bigcup_{K\in \mathcal{K}}N_K$, $x(e):=1$ for edges $e$ in $\{\{t,u\},\{s,t\},\{s,u\}\}$, $x(e):=2$ for edges $e$ in $N\cup (F\setminus \{s,u\})$, and $x(e):=0$ for the remaining edges in $G+F$. 

\end{enumerate}
Next we find the size of the set $F$ returned by the algorithm. We observe that $|F|=\Ceil{\bigl. |S|/2}$. 
It remains to bound $|S|$. For this, we count the number of vertices in the graph using the matched and exposed vertices. We have that $|V|=2|M|+|S|+\{\text{number of $M$-exposed vertices in $B_{\ge 3}$}\}$. By property ($iii$) of Theorem \ref{thm:ged-properties} and the choice of the matching $M$, we have $|V|=2|M|+|S|+2(\nu_f(G)-\nu(G))$. Since $M$ is a maximum cardinality matching, we know that $|M|=\nu(G)$ and hence, $|S|=|V|-2\nu_f(G)$. 
\end{proof}

Our next claim shows a lower bound on the optimal solution that matches the upper bound and hence proves the optimality of the returned solution. 
\begin{claim}
\label{C:sound}
Let $F'$ be a set of non-edges of $G$. Suppose $G+F'$ has a perfect $2$-matching. Then $|F'|\ge \Ceil{\bigl. |V|/2-\nu_f(G) }$.
\end{claim}
\begin{proof}
We first note that the addition of a non-edge can increase the value of the maximum fractional matching by at most one.  That is, for every graph $H$ and every non-edge $e$ of $H$, we have $\nu_f(H+e)-\nu_f(H)\le 1$ (this can be shown by considering the dual problem, namely the minimum fractional vertex cover). Now, consider an arbitrary ordering of the edges in the solution $F'$ and let $F_i'$ denote the set of first $i$ edges according to this order and let $F_0'=\emptyset$. Then, 
\[
\nu_f(G+F')-\nu_f(G)=\sum_{i=1}^{|F'|} \left(\nu_f(G+F_i)-\nu_f(G+F_{i-1})\right) \le |F'|.
\]
Thus, we have $|F'|\ge \nu_f(G+F')-\nu_f(G)$. We observe that if $G+F'$ has a perfect $2$-matching, then $\nu_f(G+F')=|V|/2$. Hence, $|F'|\ge |V|/2-\nu_f(G)$. Finally, we observe that $|F'|$ has to be an integer and hence, $|F'|\ge \Ceil{\bigl. |V|/2-\nu_f(G) }$. 
\end{proof}

%

\section{Singular Matrices} 
\label{S:Singular}

In this section, we focus on singular signings.  We will prove that \textsc{SingularSigning} is solvable for adjacency matrices of bipartite graphs (Theorem~\ref{thm:singular-singing-for-bipartite-graphs}) and that \textsc{SingularSigning} for arbitrary matrices is NP-complete. The latter result will be used to complete the proof of Theorem~\ref{thm:bounded-ev} in Section~\ref{S:hardness-evalue}.

\subsection{Finding Singular Signings of Bipartite Graphs} \label{SS:bipartite-singular-algorithm}
In this section, we characterize bipartite graphs whose signed adjacency matrix is invertible for all signings. We use this characterization to prove Theorem \ref{thm:singular-singing-for-bipartite-graphs}. 
We will use the following results by Little \cite{Little:1972} for our characterization.  
Lemma~\ref{L:bipartite-pm-parity-signing} stated below is a slight extension of the original result by Little.  We include its proof in the appendix (Section~\ref{A:bipartite-pm-parity-signing}) for the sake of completeness.



\begin{lemma}[Little \cite{Little:1972}]
\label{L:bipartite-pm-parity-signing}
Let $G$ be a graph 
with adjacency matrix $A_G$.
Then $\det(A_G(s))$ is even for all signings $s$ if and only if $G$ has an even number of perfect matchings.
\end{lemma}
%
%
%

\begin{theorem}[Little \cite{Little:1972}]
\label{L:bipartite-pm-little}
Let $G$ be a graph. Then $G$ has an even number of perfect matchings if and only if there is a set $S \subseteq V(G)$ such that 
every vertex in $G$ has even number of neighbors in $S$. 
Moreover, if $G$ has an even number of perfect matchings, then such a set $S$ can be found in polynomial time.
\end{theorem}

%

We now have the ingredients to characterize bipartite graphs whose signed adjacency matrix is invertible for all signings. 

\begin{lemma}
\label{lem:bipartite_characterization}
Let $G$ be a bipartite graph 
and let $A_G$ be the adjacency matrix of $G$.  Then $\det(A_G(s)) \neq 0$ for all signings $s$ if and only if $G$ has an odd number of perfect matchings.  
\end{lemma}

\begin{proof}
Suppose $G$ has an odd number of perfect matchings. By Lemma \ref{L:bipartite-pm-parity-signing}, we have that $\det(A_G(s)) \neq 0$ for all signings $s$.

Now suppose that $G$ has an even number of perfect matchings. By Theorem \ref{L:bipartite-pm-little}, there exists a set $S \subseteq V(G)$ such that $|N_G(v) \cap S|$ is even for all $v \in V(G)$. We observe that the subgraph $G[S]$ induced by $S$ is bipartite with every vertex having even degree. Thus, any closed walk on $G[S]$ has even number of edges and every connected component in $G[S]$ has an Eulerian tour with even number of edges. Let $C$ be a connected component of $G[S]$ with $m$ edges and let $T \coloneqq (e_1, e_2, \ldots, e_m)$ be an ordering of the edges that represents an Eulerian tour of $C$. Then we sign edge $e_i$ to be positive if $i$ is even and negative otherwise.  Every vertex $v \in V(G) \setminus S$ has even number of edges between $v$ and vertices in $S$. We partition the edges incident to $v$ into two arbitrary parts of equal size and sign all the edges in one part to be positive and the rest of the edges in the other part to be negative. Let $\hat s$ denote the resulting signing.

Under the signing $\hat s$ every vertex $v$ of $G$ has an equal number of positive and negative edges to vertices in $S$. Thus, the sum of the column vectors corresponding to the vertices in $S$ will be zero and hence $\det(A_G(\hat s)) = 0$.
\end{proof}

We note that the proof of Lemma \ref{lem:bipartite_characterization} is constructive since we can find a set $S$ for which every vertex has even number of neighbors in $S$ in polynomial time by Theorem~\ref{L:bipartite-pm-little}. Thus, Theorem \ref{thm:singular-singing-for-bipartite-graphs} follows from Theorem~\ref{L:bipartite-pm-little} and Lemma \ref{lem:bipartite_characterization}.

\subsection{Hardness of \textsc{SingularSigning}} \label{SS:hardness-singular}
 
In order to show the NP-completeness result, we reduce from the partition problem, which is a well-known NP-complete problem \cite{Ka72}. We recall the problem below:

\vspace{11pt}
\noindent {\textsc{Partition}}: Given an $n$-dimensional vector $b$ of non-negative integers, determine if there is a $\pm1$-signing vector $z$ such that the inner product $\langle b, z \rangle$ equals zero. 
\vspace{11pt}

\begin{lemma}
\label{L:singular-signing-npc}
\textsc{SingularSigning} is NP-complete.
\end{lemma}

\begin{proof}
\textsc{SingularSigning} is in NP since if there is an (off-diagonal) signing of the given matrix that is positive semi-definite or singular, then this signing gives the witness. In particular, we can verify if a given (off-diagonal) symmetric signed matrix is positive semi-definite or singular in polynomial time by computing its spectrum \cite{book-matrix-computations}. 

We show NP-hardness of \textsc{SingularSigning} by reducing from \textsc{Partition}. Let the $n$-dimensional vector $b\coloneqq ( b_1, \ldots, b_n )^T$ be the input to \textsc{Partition}, where each $b_i$ is a non-negative integer. We construct a matrix $M$ as an instance of \textsc{SingularSigning} as follows: Consider the following $(n+2) \times (n+2)$-matrix
\begin{equation*}
M \coloneqq
\begin{bmatrix}
I_n & b & \mathbf{1}_n \\
b^T & \langle b, b \rangle & 0\\
 \mathbf{1}_n^T & 0 & n \\   
\end{bmatrix},
\end{equation*}
where $I_n$ is the $n\times n$ identity matrix and $\mathbf{1}_n$ is the $n$-dimensional column vector of all ones. 
%
%
Claim \ref{claim:singular-signing-correctness} proves the correctness of the reduction to \textsc{SingularSigning}.
\end{proof}

\begin{claim}\label{claim:singular-signing-correctness}
The matrix $M$ has a symmetric off-diagonal signing $s$ such that $M(s)$ is singular if and only if there is a vector $z\in \{\pm 1\}^n$ such that the inner product $\langle b,z\rangle$ is zero. 
\end{claim}
\begin{proof}
Construct the Schur complement $M'_C$ of $C$ of $M'$ as in Claim \ref{claim:psd-signing-correctness}.
Using property $(ii)$ of Lemma \ref{lem:schur-complement-properties}, we have that 
\begin{equation*}
\det M' = \det(I_n) \cdot \det(M'_C)= \det(I_n) \cdot \det \left(
\begin{bmatrix}
0 & - \langle \hat b, z \rangle\\
- \langle \hat b,z \rangle & 0 \\    
\end{bmatrix}
\right) = - \langle \hat b, z \rangle^2.
\end{equation*}
Therefore, $\det M' = 0$ if and only if $\langle \hat b, z \rangle = 0$. We note that $\langle \hat b,z\rangle =0$ if and only if there is a $\pm 1$-vector $z'$ such that $\langle b, z'\rangle=0$. 
\end{proof}


%

\section{Hardness of Eigenvalue Problems}
\label{S:hardness-evalue}

In this section we prove that \textsc{PsdSigning} and \textsc{BoundedEvalueSigning} are NP-complete.  Together with Lemma~\ref{L:singular-signing-npc} this completes the proof of Theorem \ref{thm:bounded-ev}.

\subsection{Hardness of Positive Semi-definite Signing Problem} 
\label{SS:hardness-psd}
 
In order to show the NP-completeness of \textsc{PsdSigning}, we again reduce from \textsc{Partition} \cite{Ka72}. 
The proof has a similar outline to the NP-completeness proof of \textsc{SingularSigning} (Lemma~\ref{L:singular-signing-npc}).

\begin{lemma}\label{lem:psd}
\textsc{PsdSigning} is NP-complete.
\end{lemma}

\begin{proof}
\textsc{PsdSigning} is in NP since if there is an (off-diagonal) signing of the given matrix that is positive semi-definite, then this signing gives the witness.  In particular, we can verify if a given (off-diagonal) symmetric signed matrix is positive semi-definite in polynomial time by computing its spectrum \cite{book-matrix-computations}. 

We show NP-hardness of \textsc{PsdSigning} by reducing from \textsc{Partition}. Let the $n$-dimensional vector $b\coloneqq ( b_1, \ldots, b_n )^T$ be the input to the \textsc{Partition} problem, where each $b_i$ is a non-negative integer. We construct a matrix $M$ as an instance of \textsc{PsdSigning} as follows: Consider the following $(n+2) \times (n+2)$-matrix
\begin{equation*}
M \coloneqq
\begin{bmatrix}
I_n & b & \mathbf{1}_n \\
b^T & \langle b, b \rangle & 0\\
 \mathbf{1}_n^T & 0 & n \\   
\end{bmatrix},
\end{equation*}
where $I_n$ is the $n\times n$ identity matrix and $\mathbf{1}_n$ is the $n$-dimensional column vector of all ones. 
%
%
Claim \ref{claim:psd-signing-correctness} proves the correctness of the reduction to \textsc{PsdSigning}.
\end{proof}

\begin{claim}\label{claim:psd-signing-correctness}
The matrix $M$ has a signing $s$ such that $M(s)$ is positive semi-definite if and only if there is a $\pm 1$-vector $z$ such that the inner product $\langle b,z\rangle$ is zero. 
\end{claim}
\begin{proof}
We may assume that any signed matrix $M(s)$ that is positive semi-definite may not have negative entries in the diagonal because a positive semi-definite matrix will not have negative entries on its diagonal. Hence, we will only consider symmetric off-diagonal signing $s$ of the matrix $M$ of the following form: 
\begin{equation*}
M' \coloneqq M(s) =
\begin{bmatrix}
I_n & \hat b & z \\
\hat b^T & \langle b, b \rangle & 0\\
 z^T & 0 & n \\    
\end{bmatrix},
\end{equation*}
where the $n$-dimensional vector $z$ takes values in $\{\pm 1\}^n$ and $\hat b = (\hat b_1, \ldots, \hat b_n)^T$, where $\hat b_i$ takes value in $\{\pm b_i\}$ for every $i$.
Let 
\begin{align*}
A & := I_n,\\
B & := 
\begin{bmatrix}
\hat b & z
\end{bmatrix} \text{, and}\\
C & := 
\begin{bmatrix}
\langle b,b \rangle & 0\\
0 & n\\
\end{bmatrix}.
\end{align*}
Since $A=I_n$ is invertible, the Schur complement of $C$ in $M'$ is well-defined and is given by 
\begin{align*}
M'_C&=
\begin{bmatrix}
\langle b, b \rangle & 0\\
0 & n \\    
\end{bmatrix}
- 
\begin{bmatrix}
\hat b^T \\
 z^T \\    
\end{bmatrix}
I_n^{-1}
\begin{bmatrix}
\hat b & z
\end{bmatrix}\\
& = 
\begin{bmatrix}
\langle b, b \rangle & 0\\
0 & n \\    
\end{bmatrix}
- 
\begin{bmatrix}
 \langle \hat b, \hat b \rangle &  \langle \hat b, z \rangle \\
 \langle \hat b,z \rangle & \langle z, z \rangle \\    
\end{bmatrix}\\
&= 
\begin{bmatrix}
0 & - \langle \hat b, z \rangle\\
- \langle \hat b,z \rangle & 0 \\    
\end{bmatrix},
\end{align*}
where the last equation follows because we have $\langle \hat b, \hat b \rangle = \langle b, b \rangle$ and $\langle z,z\rangle = n$. 

We note that $A=I_n$ is positive definite. Therefore, by property $(1)$ of Lemma \ref{lem:schur-complement-properties}, the matrix $M'$ is positive semi-definite if and only if $M'_C$ is positive semi-definite. Therefore, $M'$ is positive semi-definite if and only if $\langle \hat b, z \rangle = 0$. Finally, we note that $\langle \hat b, z \rangle = 0$ if and only if there is a $\pm 1$-vector $z'$ such that $\langle b, z' \rangle=0$. 
\end{proof}
\subsection{Hardness of Bounded Eigenvalue Signing Problem} \label{SS:hardness-bounded-evalue}

To prove  that \textsc{BoundedEvalueSigning} is NP-complete, we consider the following problem that is closely related to \textsc{PsdSigning}:\\

\noindent {\textsc{NsdSigning}}: Given a real symmetric matrix $M$, verify if there exists a signing $s$ such that $M(s)$ is negative semi-definite.\\

We observe that a real symmetric $n\times n$ matrix is positive semi-definite if and only if $-M$ is negative semi-definite. Lemma \ref{lem:psd} and this observation lead to the following corollary. 

\begin{corollary}\label{cor:nsd}
\textsc{NsdSigning} is NP-complete.
\end{corollary}

We next reduce \textsc{NsdSigning} to \textsc{BoundedEvalueSigning} which also completes the proof of Theorem \ref{thm:bounded-ev}. 

\begin{lemma}
\textsc{BoundedEvalueSigning} is NP-complete.
\end{lemma}

\begin{proof}
\textsc{BoundedEvalueSigning} is in NP since if there is an off-diagonal signing of a given matrix that has all eigenvalues bounded above by a given real number $\lambda$, then this signing gives the witness. We can verify if 
all eigenvalues of a given off-diagonal symmetric signed matrix are at most $\lambda$ in polynomial time by computing the spectrum of the matrix.

We show NP-hardness of \textsc{BoundedEvalueSigning} by reducing from \textsc{NsdSigning} which is NP-complete by Corollary~\ref{cor:nsd}.  Let the real symmetric $n \times n$ matrix $M$ be the input to the \textsc{NsdSigning} problem. We construct an instance of \textsc{BoundedEvalueSigning} by considering $\lambda =0$ and the matrix $M'$ obtained from $M$ as follows (where $|a|$ denotes the magnitude of $a$):
\begin{equation*}
M'[i,j]=\left\{
\begin{array}{ll}
M[i,j] &\text{if $i\neq j$,} \\
-\Abs{\bigl. M[i,j]} &\text{if $i=j$.}
\end{array}
\right.
\end{equation*}
We observe that every negative semi-definite signing of $M$ has to necessarily have negative values on the diagonal. 
Hence, there is a signing $s$ such that that $M(s)$ is negative semi-definite if and only if there is an off-diagonal signing $t$ such that $\lambda_{\max} (M'(t)) \leq \lambda = 0$.
\end{proof}

\section{Discussion}
\label{S:conclusion}
Signed matrices are encountered in a variety of fields. Our work sheds light on the complexity landscape surrounding the spectral aspects of symmetric signings. Our results suggest that the key combinatorial structures are perfect $2$-matchings and perfect matchings in the support graph of the matrix. 

The complexities of the four problems that we studied in this work are still open and are of special interest when we restrict the input to be \emph{graph-related matrices}. In particular, are \textsc{BoundedEvalueSigning}, \textsc{PsdSigning}, and \textsc{SingularSigning} efficiently solvable or NP-complete for graph-related matrices? We show that \textsc{SingularSigning} for adjacency matrices of bipartite graphs can be solved efficiently. 
It would be interesting to solve \textsc{SingularSigning} for adjacency matrices of arbitrary graphs.

The search variant of the four problems for graph-related matrices are also of interest. In particular, we believe that our search technique for finding an invertible signing for the adjacency matrix of bipartite graphs merits careful investigation: we initialize with a signed subgraph satisfying the required spectral property and incrementally sign and add edges without violating the spectral property (Lemma \ref{lem:partial-signing-bipartite}). Variations of this technique could be applicable to solve the search variant of all four problems. It would be especially interesting from the perspective of expander construction if the technique can be used to solve the search variant of \textsc{BoundedEvalueSigning} for the adjacency matrix of a given $d$-regular bipartite graph when $\lambda=2\sqrt{d-1}$. 

The counting variant of the four problems for graph-related matrices asks for the number of signed matrices which satisfy the spectral property associated with the problem. If we could obtain a lower bound on the number of such signed matrices, then we could design randomized algorithms to solve the search problem (pick a random signing and succeed with probability equal to the fraction of signed matrices satisfying the required spectral property). While we have shown a lower bound on the number of invertible signed adjacency matrices, it remains open to address the remaining three spectral properties for graph-related matrices. For instance, experimental evidence suggests that the fraction of \emph{Ramanujan signings} of a $d$-regular bipartite graph is at least a constant and it has been a challenging open problem to prove (or disprove) this. 
Furthermore, counting guarantees for the number of invertible signings could also help to derandomize or decrease the amount of randomness in the algebraic method for verifying the existence of a perfect matching in a given bipartite graph. 

\paragraph{Acknowledgements.}
The authors would like to thank the anonymous referees for their helpful comments.  

\bibliographystyle{newabuser}
\bibliography{references}

\appendix 
\section{Appendix}
\subsection{Proof of Lemma \ref{L:bipartite-pm-parity-signing}}
\label{A:bipartite-pm-parity-signing}

\begin{proof}[Proof of Lemma \ref{L:bipartite-pm-parity-signing}]
Suppose $G$ has $n$ vertices. We recall that the permutation expansion of the determinant of $A_G(s)$ for a signing $s$ is given by 
\[
\det A_G(s) = \sum_{\sigma\in S_n} \sgn(\sigma) \cdot \prod_{i=1}^n A_G(s)[i,\sigma(i)].
\]
In order to prove the claim, it suffices to show that every perfect 2-matching in $G$ that is not a perfect matching must contribute an even number of distinct nonzero terms of the same sign to the determinant of $A_G(s)$ for all signings $s$. We show this next.

Let $M$ be a perfect $2$-matching in $G$ that is not a perfect matching. Let $C$ be the set of vertex disjoint cycles in $M$. Then $M$ corresponds to $2^{|C|}$ distinct permutations of $S_n$ and thus $2^{|C|}$ distinct terms in the permutation expansion of the determinant of $A_G(s)$ for all signings $s$. Since all such permutations have the same cycle structure it follows that each term must have the same sign. Thus each perfect $2$-matching that is not a perfect matching in $G$ contributes an even number to the permutation expansion of the determinant of $A_G(s)$.
\end{proof}

\end{document}